\documentclass{article}

\usepackage{PRIMEarxiv}

\usepackage[utf8]{inputenc} 
\usepackage[T1]{fontenc}    
\usepackage[hidelinks]{hyperref}       
\usepackage{url}            
\usepackage{booktabs}       
\usepackage{amsfonts}       
\usepackage{amsmath}
\usepackage{amsthm}

\newtheorem{thm}{Theorem}
\newtheorem{lemma}[thm]{Lemma}

\usepackage{amssymb}
\usepackage{nicefrac}       
\usepackage{microtype}      
\usepackage{lipsum}
\usepackage{fancyhdr}       
\usepackage{graphicx}       

\usepackage{ragged2e}
\usepackage{mathptmx}

\usepackage{subcaption}
\usepackage{floatrow}
\DeclareCaptionLabelFormat{andtable}{#1~#2  \&  \tablename~\thetable}
\newfloatcommand{capbtabbox}{table}[][\FBwidth]

\usepackage{algpseudocode}
\usepackage{algorithm}
\usepackage{tabularray}
\usepackage{textcomp}
  
\title{A Greedy Annealing-based Route Generation Algorithm}

\author{
  Jordan Makansi \\
  Independent Consultant
   \And
  David E. Bernal Neira \\
  Davidson School of Chemical Engineering, Purdue University
}

\begin{document}
\maketitle

\raggedbottom

\begin{abstract}
Routing and scheduling problems with time windows have long been important optimization problems for logistics and planning.
Many classical heuristics and exact methods exist for such problems.
However, there are no satisfactory methods for generating routes using quantum computing (QC), mainly because of two reasons: inequality constraints and the trade-off of feasibility and solution quality.
Inequality constraints are typically handled using slack variables, and samples are filtered to find feasible solutions.
These challenges are amplified in the presence of noise inherent in QC.
Here, we propose a greedy algorithm for generating routes using information from all samples obtained from the quantum computer.
By noticing the relationship between qubits in our formulation as a directed acyclic graph (DAG), we designed an algorithm that adaptively constructs a feasible solution. 
We prove its convergence to a feasible solution (without post-processing) and an optimal solution if an exact solver solves the QUBO subproblem.
We demonstrate its efficacy by solving the Fleet Sizing Vehicle Routing Problem with Time Windows (FSVRPTW).
Our computational results show that this method obtains a lower objective value than the current state-of-the-art annealing approaches, both classical and hybrid quantum-classical, for the same amount of time using D-Wave’s Hybrid Solvers. 
We also show its robustness to noise on D-Wave’s \texttt{Advantage2} through computational results as compared to the standard filtering approach on \texttt{DWaveSampler}.
\end{abstract}


\section{Introduction}
One of the most important problems in operations research is servicing customers within a certain time window. 
The simplest version of this is the traveling salesperson problem with time windows. 
In this paper, we study the Fleet Size Vehicle Routing Problem with Time Windows (FSVRPTW), which determines the minimum number of vehicles required to service all customers within their respective time windows.
The FSVRPTW problem is practically relevant on its own, but it is also a subproblem in many solution methods for more complex routing problems, such as column generation, Bender's decomposition, branch-and-cut (B\&C), and branch-and-price (B\&P). 
However, in the aforementioned solution approaches, an exact solution method to the subproblem is not always necessary - but a computationally tractable one is. Therefore, the community would benefit from faster solutions to this problem. 

Quantum computing (QC) offers exciting possibilities for performing specific computational tasks better than classical solvers and may be able to speed up challenging optimization problems such as route planning.
Specific algorithms have proven advantage over their classical counterparts in fault-tolerant conditions, such Grover's algorithm for search over unstructured databases~\cite{grover1996fast} proved to improve over the best possible classical alternative and Shor's factorization of integers into prime factors~\cite{shor1999polynomial}, better than any known classical method.
In the near future, it is not certain whether a quantum advantage can be observed due to the noise and current size capabilities of quantum computers. As such, many hybrid approaches involving classical and quantum computing have been proposed. However, it remains relevant to determine the best quantum approach to certain problems so that when the QC capabilities become more advanced, the solution methods are readily available.

Routing problems on quantum computers have been solved since Feld et al. \cite{QVRP-8}. However, no annealing-based methods have shown to be practical for generating routes themselves. Existing approaches that have shown to be capable of solving problems at scale generate routes classically and solve a set cover problem. On the other hand, the existing annealing-based approaches do not scale well due to slack variables and other inefficient encoding schemes \cite{9781399}.
Here we review them:
\cite{QVRP-4} solves the traveling salesperson problem with time windows (TSPTW) using slack variables and one-hot encoding, which makes the number of qubits impractical. 
\cite{QVRP-1} uses position-based indexing. While this works under the single-agent setting in TSPTW, it loses much in optimality in a multi-agent setting because it assumes each agent makes a maximum number of stops. 
\cite{QVRP-2} creatively discretizes time using a state-based description, but the formulation of time windows has a discrepancy with the capacity and state description. 
\cite{QVRP-7} presents several QUBO formulations but only solves small-sized problems and, aside from the route-based formulation, has difficulty obtaining feasible solutions. 
\cite{QVRP-6}, \cite{QVRP-11}, and \cite{QVRP-10} all present creative hybrid solution methods but still generate routes classically.

The majority of existing QC solutions to Sherrington-Kirkpatrick (SK) Ising models involve sampling the QC a number of times and eventually selecting the sample with the lowest energy. However, emerging research has shown the potential to use all of the samples from the QC to take advantage of the correlation between samples \cite{QC-15}. We review them here.
\cite{QC-11} uses an iterative approach that fixes qubits based on their one-body expectation value, computed by using all of the samples returned from the annealer. Qubits with expectation values above a certain threshold are fixed to a particular value based on a majority vote, and a smaller SK Ising model is solved at the next iteration. The expectation value, in some sense, measures how certain we are of its value.
\cite{QC-12}, unlike \cite{QC-11}, has two separate procedures for selecting qubits to fix and their values. After selecting qubits, their values are determined by brute-force evaluation of the expectation value of the objective function, over all possible values of the selected variables. 
\cite{QC-13} uses two phases instead of an iterative approach: 1) Run QAOA. 2) Compute the correlation matrix given by the output from sampling the QC circuit constructed from QAOA. The values of the principal eigenvector are rounded to obtain a final solution. 
Other efforts to improve solution quality, e.g., unbalanced penalization \cite{QAPP-3}, still cannot guarantee a feasible solution and require heavy classical pre-processing for tuning penalty parameters. They also do not disclose computation time in their computational results. 

However, these approaches do not directly apply to \textit{constrained} optimization problems; fixing variables may decrease the expectation value of the energy, but it might also violate constraints. The fixed values are maintained permanently throughout the remainder of the algorithm. For constrained problems, we require a delicate approach to ensure a feasible solution can still be returned after the variables are fixed. We propose an algorithm that takes advantage of both emerging research to utilize the correlation between samples and is guaranteed to return a feasible solution to the FSVRPTW.

In this paper, we present an annealing-based algorithm for generating routes under time window constraints that is suitable for the NISQ era \cite{Preskill_2018} by its 
convergence to a feasible solution, and competitive computation time and solution quality.

Our main contributions are as follows:

\begin{enumerate}
\item An annealing-based algorithm for generating routes that represents the samples returned by the annealer as a DAG. 
\item Prove that it converges to a feasible solution and, if an exact solver solves the subproblem, converges to the optimal solution. 
\item Show computational results on D-Wave, benchmarking it against classical, hybrid, and quantum annealing-based approaches based on computation time and optimality for a variety of practical-sized benchmark instances.
\end{enumerate}

\section{Preliminaries} \label{sec:prelim}
We now review the solutions and technologies used throughout this work.

\textbf{Mixed Integer Programming (MIP) Solver}:  An MIP solver will solve a Mixed Integer Program exactly to obtain an optimal solution using Branch-and-Bound, which has exponential complexity in worst case \cite{gurobi, morrison2016branch}.  A generic MIP is shown in equation \ref{eq:mip}.  In this work the MIP solver \texttt{Gurobi 11.0.1} was used used to obtain the optimal solution to Problem \eqref{eq:qubo_fsvrptw}.

\begin{subequations} \label{eq:mip}
\begin{alignat}{2}
& \min_{x,y} & \mathbf{c}^\top \mathbf{x} + \mathbf{d}^\top \mathbf{y}
\\
& \text{s.t.} &\mathbf{A} \mathbf{x} + \mathbf{B} \mathbf{y} \leq \mathbf{b}
\\
&& \mathbf{x} \in \mathbb{R}
\\
&& \mathbf{y} \in \mathbb{Z}
\end{alignat}
\end{subequations}
\textbf{Quadratic Unconstrained Binary Optimization (QUBO)}.
Many discrete optimization problems can be formulated as a Quadratic Unconstrained Binary Optimization (QUBO) problem. It is the common format used to input optimization problems into annealers. Given square matrix $Q \in \mathbb{R}^{n \times n}$, the problem is to minimize $x^T Q x$ over binary variables $x_i \in \{0,1\}^n$: 
\begin{equation} \label{qubo}
\min_{x\in\{0,1\}^n} x^T Q x 
\end{equation}
Coefficients for linear terms in an optimization problem may be represented in the QUBO form by the diagonal elements of the matrix $Q$ since $x^2_i=x_i$. 
To represent this optimization problem on a quantum device, the QUBO is transformed into an Ising model with the Hamiltonian consisting of a weighted sum of products of spin variables \cite{Kirkpatrick1983SimulatedAnnealing}.  We use $Z_i$ to denote the spin variables. We may convert binary variables $x_i$ into spin variables $Z_i \in \{-1,1\}$ by $x_i = \frac{Z_i +1}{2}$ and represent them in an Ising model of the following Hamiltonian 
\begin{equation} \label{ising}
\mathcal{H} := \sum_{i=1} v_i Z_i + \sum_{j<i} J_{ij} Z_i Z_j + C
\end{equation}

where $J_{ij}, v_i$ encode the elements in the matrix $Q$. $C$ is a constant offset, which is typically ignored since it does not affect the optimal objective value. There is a one-to-one correspondence between the spin variables $Z_i$ and the binary variables $x_i$. Once a problem is cast as a QUBO, it may be solved on an annealer. A quantum annealer works by exploiting quantum mechanical effects to find low energy states of the Hamiltonian in Eq. \eqref{ising} \cite{morita2008mathematical}. The output of an annealer is non-deterministic, so samples are used to obtain a representative distribution. Given $M$ samples, the expected value of a variable $Z_i$ is $\langle Z_i \rangle=\sum^M_m Z^m_i/M$, where $m$ is the index of the sample. 

\textbf{Quantum Annealer (QPU)}:  A QPU uses adiabatic evolution to solve QUBOs by gradually evolving from an easy-to-solve initial state to the problem's final configuration, utilizing quantum tunneling to avoid local minima \cite{DWaveQPUAnnealing2024}.  We used D-Wave's \texttt{Advantage} solver with \texttt{num\_reads}=1000.

\textbf{Simulated Annealing Sampler (SAS)}:  A classical metaheuristic algorithm that mimics the physical process of annealing.  It leverages the Markov Chain Monte Carlow method to explore the solution space. The probability of accepting a new solution is based on the Gibbs distribution.  This probability is decreased according to an \textit{annealing schedule}.  We use the \texttt{SimulatedAnnealingSampler} (SAS) provided by D-Wave, and we use the default settings \cite{DWavesysDWavenealx2014}.

\textbf{LeapHybridCQMSampler (CQM)}:  A hybrid classical/quantum annealer that is specifically designed to solve problems with constraints. It combines several classical heuristics, such as Simulated Annealing, Tabu Search, and others, with Quantum Annealing to find a solution \cite{DWaveHybridSolvers2021}.  Its algorithmic parameters include \texttt{num\_reads}, which is the number of samples returned by the annealer. The more samples, the greater the chance of finding a solution with low energy. We set \texttt{num\_reads}=1000.

\section{Problem Formulation}
We are given a timetable of customers, each with a time window $[e_i, l_i]$ and a service time $q_i$; an example is shown in Tab. \ref{fig:Fig_1}. The start of service must begin within the time window.

\begin{figure}
\begin{tabular}[b]{@{}llll@{}}
\multicolumn{3}{c}{Timetable}\\
\toprule
i & $e_i$ & $l_i$ & $q_i$\\ \midrule
0 & 0.0 & $\infty $ & 0.0 \\
1 & 1.0& 1.15 & 1.0 \\
2 & 3.5& 3.75 & 1.0 \\
N & 0.0& $\infty $ & 0.0 \\ \bottomrule
\end{tabular}
\qquad
\begin{tabular}[b]{@{}lllll@{}}
\multicolumn{5}{c}{Dist. Matrix}\\
\toprule
i & 0 & 1 & 2 & N \\ \midrule
0 & 0 & 1 & 3 & 0 \\
1 & 1 & 0 & 2 & 1 \\
2 & 3 & 2 & 0 & 3 \\
N & 0 & 1 & 3 & 0 \\ \bottomrule
\end{tabular}
\qquad
    \centering
  \includegraphics[width=0.52\textwidth]{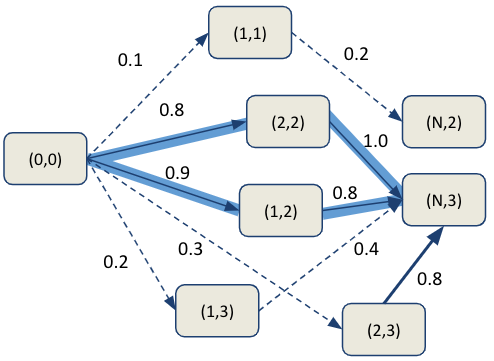}
    \captionlistentry[table]{A table beside a figure}
    \captionsetup{labelformat=andtable}
    \caption{Timetable and distance matrix for $N$=2 example. Initial set of variables $X^0$ represented as a DAG (denoted DAG($X^0$)), continued from example shown in the Table, labeled with possible one-body expectation values from step \ref{exp_end} in Alg. \ref{alg:greedy}. Variables selected in step \ref{select} in Alg. \ref{alg:greedy}, using $\theta = 0.5$, are marked with solid arrows. Paths found in step \ref{find_paths} in Alg. \ref{alg:greedy} are highlighted in blue.}
    \label{fig:Fig_1}
\end{figure} 
We start by discretizing time such that for each time window, there exists a time discretization in $[e_i, l_i]$. If an agent arrives at a customer before $e_i$, we allow for waiting. We create variables $x_{i,s,j,t}$ representing ``an agent leaves customer $i$ at time $s$ and subsequently leaves customer $j$ at time $t$''. We duplicate the depot to create the original depot $0$ and the final depot $N$. The set $T$ is the set of time discretizations, and the set $W$ is the set of customers, which excludes depots. For any variable $x_{i,s,j,t}$ we can compute $b_{ij} = \max\{e_j, s + d_{ij}\}$where $d_{ij}$ is the distance between customers $i$ and $j$. $b_{ij}$ is the earliest possible time we can start servicing customer $j$, having left customer $i$ at time $s$. We use an arc-based formulation from \cite{QVRP-7}: the constrained binary optimization is shown in Eq. \eqref{eq:qubo_fsvrptw}.

The objective is in \eqref{eq:obj_qubo_fsvrptw}. Constraint \eqref{eq:qubo_cover} ensures that all customers are visited. Constraint \eqref{eq:flow} enforces flow between consecutive customers. The remaining constraints are pre-processing and domain constraints.
\begin{subequations} \label{eq:qubo_fsvrptw}
\begin{alignat}{2}
& \min_{x_{i,s,j,t} \in \{0,1\} } && \sum_{j\in W} \sum_t x_{0,0,j,t} \label{eq:obj_qubo_fsvrptw}
\\
& & & \text{s.t. } \sum_{i\in\{0,W\},s,t} x_{i,s,j,t} = 1 , \,\, \forall j \in W \label{eq:qubo_cover}
\\
& & & \sum_{j\in\{0,W\},j\neq i, t} x_{j,t,i,s} = \sum_{j\in\{W, N\},j\neq i, t} x_{i,s,j,t}, \, \forall i \in W, \forall s \label{eq:flow}
\\ 
& & & x_{i,s,0,t} = 0, \, \forall s,t,\, \forall i \in \{0, W, N\} \label{eq:no_enter_original_depot}
\\ 
& & & x_{N,s,j,t} = 0, \, \forall s,t , \, \forall j \in \{0, W, N\} \label{eq:no_leave_final_depot}
\\ 
& & & x_{0,s,j,t} = 0, \, \forall s \neq 0 , \forall t, \forall j \in \{0, W, N\} \label{eq:leave_original_time_0}
\\ 
& & & x_{i,s,j,t} = 0, \, \forall i,s,t, \forall j \in \{0, W, N\}: b_{ij} > l_j, \label{eq:time_consistent_late}
\\ 
& & & x_{i,s,j,t} = 0, \, \forall i,s, t, \forall j \in \{0, W, N\} : t < b_{ij} + q_j \label{eq:time_consistent_early}
\\ 
& & & x_{i,s,i,t} = 0, \, \forall s, t, \forall i \in \{0, W, N\} \label{eq:no_self_loop}
\\ 
& & & x_{i,0,j,t} = 0, \, \forall i \neq 0, \forall t, \forall j \in \{0, W, N\} \label{eq:no_leave_cust_at_time_0}
\\ 
& & & x_{i,s,j,0} = 0, \, \forall i, s, \forall j \in \{0, W, N\} \label{eq:no_leave_final_depot_time_0}
\\ 
& & & x_{0,s,N,t} = 0, \, \forall s,t \label{inv_cust}.
\end{alignat}
\end{subequations}

\subsection{Pre-processing} \label{sec:pre_proc}
Our formulation initially may appear to have a formidable number of variables. However, by making a few observations, we are able to prune many of the variables by fixing their values to $0$ and removing them from the model. The variables we can prune correspond to constraints (\ref{eq:no_enter_original_depot})-(\ref{inv_cust}) and are described in Table \ref{tab:prune}.

\begin{table} 
\caption{Constraint and Description for Pre-processing variables in Problem \eqref{eq:qubo_fsvrptw}}
\label{tab:prune}
\begin{tabular}{@{}ll@{}}
\multicolumn{2}{c}{}\\
\toprule
Constraint & Description \\ \midrule
(\ref{eq:no_enter_original_depot}) & Cannot enter original depot \\ 
(\ref{eq:no_leave_final_depot}) & Cannot leave final depot \\ 
(\ref{eq:leave_original_time_0}) & Can only leave the original depot at time 0 \\ 
(\ref{eq:time_consistent_late}) & Cannot be late to service customer $j$ \\ 
(\ref{eq:time_consistent_early}) & Cannot leave customer $j$ before end of service time \\ 
(\ref{eq:no_self_loop}) & No self-loops \\ 
(\ref{eq:no_leave_cust_at_time_0}) & Cannot leave customers at time 0 \\ 
(\ref{eq:no_leave_final_depot_time_0}) & Cannot leave any nodes at time 0, after node $i$ \\ 
(\ref{inv_cust}) & Cannot travel immediately from original depot to final depot \\ \bottomrule
\end{tabular}
\end{table}

\section{Methods}
\subsection{Greedy Quantum Route-Generation Algorithm}
\label{sec:greedy}
Our algorithm for solving Problem~\eqref{eq:qubo_fsvrptw} is shown in Alg. \ref{alg:greedy}. We first observe that any subset of the variables in our formulation can be represented as a directed acyclic graph. By combining this with the samples obtained from the annealer, we can converge quickly to a feasible solution. The main idea is the following: the annealer recommends a set of variables to be used in constructing a DAG, and a classical routine is used to find feasible sub-paths within the DAG. The variables corresponding to customers along each of the paths are pruned. Variables that haven't been pruned are still ``active''. At each iteration, the set of active variables, $X$, gets smaller, while the cumulative set of feasible paths returned $S$ gets larger. The algorithm terminates when $X^l=\emptyset$, although, in practice, it may be terminated when the number of active variables is small enough to solve the problem exactly.

We use $Q$ to denote the new set of edge-and-vertex-disjoint paths found in the DAG and $P$ to denote a single path in $Q$. Paths consist of a sequence of tuples $\{(i^P_k, s_k)\}_{k=0}^{|P|}$. We use $i^P_k$ to denote the customer at index $k$ of path $P$, and no superscript is used when clear from the context.
Let $X^l$ denote the set of variables ``active'' at iteration $l$, $S^l$ denote the cumulative set of paths found in iterations prior to $l$, and $Q^l$ the set of paths found at iteration $l$. We use (\ref{eq:qubo_cover}, $i$) to denote the coverage constraints for customer $i$, and (\ref{eq:flow}, $i,s$) to denote the flow constraints for customer $i$ at timestep $s$. DAG($X$) denotes the DAG formed by the variables in set $X$. QUBO($X$) denotes solving Problem \eqref{eq:qubo_fsvrptw} using annealing over the set of variables $X$. $\mathcal{X}$ denotes the samples returned by the annealer. $A$ and $M$ are the annealing time and number of samples, respectively. They are implicitly inputs to QUBO($X$). When clear from context, $i$ is the index of a single variable in $X^l$. We use $R^\frown R'$ to denote the concatenation of two paths $R$ and $R'$. The final set of paths returned as a solution to the FSVRPTW \eqref{eq:qubo_fsvrptw} is $S$.
\begin{algorithm}
 \caption{Greedy Quantum Route-Generation Algorithm}\label{alg:greedy}
 \begin{algorithmic}[1]
 \Require $X^0=\{x_{i,s,j,t}\}$ initially, all variables ``active'', $\theta$, $M$, number of annealing samples, $A$, annealing time. 
 \Ensure $S,\text{paths found}$
 \State{$S^0 \gets \emptyset $} 
 \While{$X^l \neq \emptyset$}
 \State{$\mathcal{X}^l \gets$ QUBO($X^l$)} \label{step:annealing}
 \For{$i \in 1....|X^l|$} 
 \State{Compute $\frac{\langle Z_i \rangle + 1}{2}$ over all $M$ samples} \label{exp1}
 \EndFor \label{exp_end}
 \State $\bar{X}^l, \gets$ Select $\theta$ fraction of $X^l$ with highest $\frac{\langle Z_i \rangle + 1}{2}$ \label{select}
 \State $Q^l \gets$ Find paths in DAG($\bar{X}^l$) \label{find_paths} 
 \State $S^{l+1},X^{l+1} \gets $ \Call{Prune}{$S^l,X^l,Q^l$} \label{prune}
 \EndWhile
 \end{algorithmic}
\end{algorithm}
Step \ref{exp_end} computes a statistic used in step \ref{select} to determine which variables will participate in the DAG for finding paths. In our algorithm, we used one-body expectation values, although two-body expectation values may be used \cite{QC-12}. Parameter $\theta \in (0,1)$ is given by the user to determine the variables in $X^l$ to select at each iteration. $\theta$ may be provided as a portion of the variables to choose or as a threshold for expectation values. In our running example and our analysis, we specify $\theta$ as a portion.
For step \ref{find_paths}, any sub-routine may be used that guarantees that no two paths returned visit the same customer.  Two examples are: returning the single longest path or returning multiple paths. An iterative procedure may return multiple paths: 1. Select the longest path in the DAG. 2. Remove all nodes (and their adjacent edges) in the DAG involving any of the customers along the currently selected longest path (excluding depots $0, N$), 3. Remove the nodes and arcs along the currently selected longest path from the DAG. 4. Repeat until the DAG is empty. In the simulation results, we returned multiple paths using this procedure. Heuristics such as returning a single longest path are more conservative and allow more flexibility for the algorithm to adapt in consecutive iterations but may lead to longer computation time.
\vspace{-0.5cm}
\begin{algorithm}
\caption{\textproc{Prune}}\label{alg:prune}
\begin{algorithmic}[1]
 \renewcommand{\algorithmicrequire}{\textbf{Input:}}
 \renewcommand{\algorithmicensure}{\textbf{Output:}}
\Require $S$, set of paths found previously, $Q$, new set of paths found, $X$, current set of active variables.
\Ensure $S'$, new set of paths found, $X'$, new set of active variables, QUBO$'$, new Problem \eqref{eq:qubo_fsvrptw} formulation
 \For {$P \in Q$} \label{exp2}
 \State QUBO $\gets \Call{Concat}{S,P,\text{QUBO}}$\label{step:concat} 
 \For {$(i_k, s_k) \in P$}
 \State $X' \gets $ Apply pruning rules in \S \ref{sec:pruning_rules}. \label{step:prune}
 \If {$k \notin \{0,|P|\}$} 
 \State QUBO $\gets$ remove constraint (\ref{eq:flow},$i_k,s_k$) \label{step:remove_flow}
 \EndIf
 \If {$k\neq0$}
 \State QUBO $\gets$ remove constraint (\ref{eq:qubo_cover},$i_k$) \label{step:remove_coverage}
 \EndIf
 \EndFor
 \EndFor
\end{algorithmic}
\end{algorithm}
\subsection{DAG Representation}
Any set of active variables $X^l$ may be represented as follows: A single node is created for each unique tuple $(i,s)$, where $i \in \{0, W, N\}$ and $s \in T$. A directed arc is present from $(i,s)$ to $(j,t)$ for each variable $x_{i,s,j,t} \in X^l$.
By constraints \eqref{eq:time_consistent_late} and \eqref{eq:time_consistent_early}, after the pre-processing described in \S \ref{sec:pre_proc} variables only remain active if $s<t$. There cannot exist a cycle; otherwise, it would imply that there exists a variable $x_{i,s,j,t}$ such that $s \geq t$. DAG($X^0$) for our running example is shown in Fig. \ref{fig:Fig_1}.
\subsection{Pruning} 
\label{sec:pruning}
The subroutine \Call{Prune}{$S, X, Q$} is shown in Alg. \ref{alg:prune}.
Notice that the constraints (\ref{eq:qubo_cover}) and (\ref{eq:flow}) are separable by customer $j$, and tuples $(i,s)$ respectively. So we may prune variables and remove constraints independently for each tuple $(i,s)$, along a path in steps \ref{step:prune}, \ref{step:remove_flow}, and \ref{step:remove_coverage} in Alg. \ref{alg:prune}.
If a path is found that shares an endpoint (either start or end) of a path in $S$, the paths are concatenated by \Call{Concat}{} and the appropriate constraints are removed, as shown in Alg. \ref{alg:concat}. 
\subsubsection{Pruning Rules} \label{sec:pruning_rules}
Throughout this section, the phrase ``prune'' means the variable is set to the value specified and removed from the set of active variables $X^l$. 
For customer $i_k$ of tuple $(i_k, s_k)$, variables involving customer $i_k$ may be classified as:
\begin{enumerate}
\item $i=i_k$: "outgoing" variables.
\item $j=i_k$: "incoming" variables.
\end{enumerate}
Index $k \in 0,...|P|$ along a path $P$ may be classified
\begin{enumerate}
\item $k \notin \{0,|P|\}$: ``interior'' indices.
\item $k \in \{0,|P|\}$: ``exterior'' indices.
\end{enumerate}
\begin{algorithm}
\caption{\textproc{Concat}}\label{alg:concat}
\begin{algorithmic}[1]
 \renewcommand{\algorithmicrequire}{\textbf{Input:}}
 \renewcommand{\algorithmicensure}{\textbf{Output:}}
\Require $S$, $P$, QUBO
\Ensure QUBO$'$, $S$
 \For {$P' \in S$} \label{exp3}
 \If {$i^P_0 = i^{P'}_{|P'|}$ for $(i^P_0,s) \in P$} 
 \State $P' \gets P'\,^\frown P$ \label{join_at_end}
 \State QUBO$\gets$Remove constraints (\ref{eq:flow}, $i^P_0,s$)
 \EndIf
 \If {$i^P_{|P|} = i^{P'}_0$ for $(i^P_{|P|}, s) \in P$} 
 \State $P' \gets P ^\frown P'$ \label{join_at_start}
 \State QUBO$\gets$Remove constraints (\ref{eq:flow}, $i^P_{|P|},s$), (\ref{eq:qubo_cover},$i^P_{|P|}$)
 \EndIf
 \EndFor
\end{algorithmic}
\end{algorithm}
\begin{figure}[h]
 \begin{subfigure}{0.45\textwidth}
 \includegraphics[width=\linewidth]{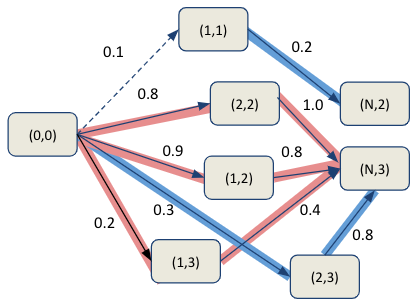}
 \caption{Arcs in paths found in step \ref{find_paths} of Algorithm \ref{alg:greedy}  are highlighted blue, and variables along the path are set to $1$ and pruned.  Arcs highlighted in red are set to $0$ and pruned.  Paths found are $P_1=\{(1,1),(N,2)\}$, $P_2=\{(0,0),(2,3),(N,3)\}$.} \label{fig:pruning_example_1}
 \end{subfigure}%
 \vspace*{\fill} 
 \begin{subfigure}{0.425\textwidth}
 \includegraphics[width=\linewidth]{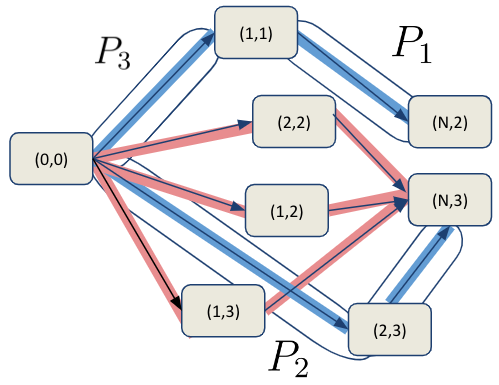}
 \caption{After pruning in the previous iteration, only variable $x_{0,0,1,1}$ remains active. Path $P_3=\{(0,0),(1,1)\}$ is found in step \ref{find_paths} of Alg. \ref{alg:greedy}, and joined with path $P_1$ in step \ref{join_at_start} of Alg. \ref{alg:concat}.} \label{fig:pruning_example_3}
 \end{subfigure}%
\caption{Pruning example for $\theta$=0.9, and subsequent iteration of Alg. \ref{alg:greedy}}
\end{figure}
\begin{table}[H]
\caption{Pruning Rules Applied to Fig. \ref{fig:pruning_example_1}}
\label{tab:pruning_rules}
\centering
\begin{tblr}{
 cell{2}{1} = {r=6}{},
 cell{2}{2} = {r=5}{},
 cell{2}{3} = {r=5}{},
 cell{8}{1} = {r=4}{},
 cell{9}{2} = {r=2}{},
 cell{9}{3} = {r=2}{},
 hline{1,14} = {-}{0.08em},
 hline{2,8,12} = {-}{},
 hline{7,9,11} = {2-5}{},
}
Path & $k$ & $(i_k, s_k)$ & Variables Pruned & Pruning Rule \\
$P_1$ & $0$ & $(1,1)$ & $x_{1,1,N,2}=1$ & Along path \\
 & & & $x_{0,0,1,2}=0$ & exterior \ref{exterior_first_not_depot} \\
 & & & $x_{0,0,1,3}=0$ & exterior \ref{exterior_first_not_depot} \\
 & & & $x_{1,3,N,3}=0$ & exterior \ref{exterior_first_not_depot} \\
 & & & $x_{1,2,N,3}=0$ & exterior \ref{exterior_first_not_depot} \\
 & $1$ & $(N,2)$ & None & exterior \ref{exterior_last_depot} \\ 
$P_2$ & $0$ & $(0,0)$ & $x_{0,0,2,3}=1$ & Along path \\
 & $1$ & $(2,3)$ & $x_{2,3,N,3}=1$ & Along path \\
 & & & $x_{0,0,2,2}=0$ & interior \ref{incoming_final_depot} \\ 
 & $2$ & $(N,3)$ & $x_{2,2,N,3}=0$ & interior \ref{outgoing_diff_s} \\
\end{tblr}
\end{table}
Using this terminology we iterate over tuples $(i_k, s_k)$ and prune by the following rules:

\textbf{Along path}: If variable $x_{i_k, s_k, i_{k+1}, s_{k+1}}$ is an arc along the path for two consecutive tuples $(i_k, s_k), (i_{k+1}, s_{k+1})$, it is set to $1$ and pruned.

If $k$ is \textbf{``interior''}:
\begin{enumerate}
\item ``outgoing'' variables are set to $0$ and pruned if either:\label{outgoing}
\begin{enumerate}
\item $i = i_k$ and $s \neq s_k$ OR \label{outgoing_diff_s} 
\item $i=i_k$, and $s=s_k$, and either:
\begin{enumerate}
\item $j \neq i_{k+1}$OR
\item $t \neq s_{k+1}$
\end{enumerate}
\end{enumerate}
\item ``incoming'' variables are set to $0$ and pruned if either:\label{incoming}
\begin{enumerate}
\item $j = i_k$ and $t \neq s_k$ OR \label{incoming_final_depot}
\item $j=i_k$, and $t=s_k$, and either:\label{incoming_not_final_depot}
\begin{enumerate}
\item $i\neq i_{k-1}$ OR
\item $t \neq s_{k-1}$
\end{enumerate}
\end{enumerate}
\end{enumerate}
If $k$ is \textbf{``exterior''}:
\begin{enumerate}
\item $k=0$:
\begin{enumerate}
\item $i_k =0$: Only prune and set $=1$ if $x_{i_k, s_k, i_{k+1}, s_{k+1}}$ is \textbf{Along path}. (Variables incoming to the depot were pruned by (\ref{eq:no_enter_original_depot})). \label{exterior_first_depot}
\item $i_k \neq 0$: Prune using rules ``interior'' \ref{outgoing} and \ref{incoming_final_depot}. \label{exterior_first_not_depot}
\end{enumerate}
\item $k=|P|$:
\begin{enumerate}
\item $i_k =N$: Prune using ``interior-incoming'' \ref{incoming_final_depot}. (We do not prune according to rules \ref{incoming_not_final_depot} because incoming arcs to final depot $N$ at different times are allowed, as shown in Fig. \ref{fig:Fig_1}). \label{exterior_last_depot}
\item $i_k \neq N$: Prune incoming variables under interior rule \ref{incoming}, and outgoing variables under interior rule \ref{outgoing_diff_s}.\label{exterior_last_not_depot}
\end{enumerate}
\end{enumerate}

\section{Analysis} \label{sec:analysis}

\begin{lemma} \label{even_prune}
In any iteration $l$, if a customer $i^P_k$ is in the interior of path $P \in S^l$, all outgoing and incoming variables to that customer have been pruned.
\end{lemma}
\begin{proof} Assume without loss of generality that $P$ was formed by the union of paths $R$ and $R'$, found at iterations $l^R$ and $l^{R'}$, respectively. $P=R ^\frown R'$. For customer $i^P_k$ there are two cases:
\begin{enumerate}
 \item $i^P_k$ was in the interior of either path $R$ or $R'$. So it was either pruned to $1$ by the ``Along path'' rules or pruned to $0$ by the ``interior'' rules in \S \ref{sec:pruning}.
 \item $i^P_k$ was in the exterior of of either path $R$ or $R'$. The incoming variables to $i^P_k$ were pruned in iteration $l^R$ by ``exterior'' rule \ref{exterior_last_not_depot} and the outgoing variables were pruned in iteration $l^{R'}$ by ``exterior'' rule \ref{exterior_first_not_depot}.
\end{enumerate}
\end{proof}
\begin{lemma} \label{cust_appear}
For every customer $i'$, $\exists$ iteration $l'$ in which the paths $Q^{l'}$ returned in step \ref{find_paths} of Alg. \ref{alg:greedy} will contain a path with a tuple $(i',s)$ for some $s$.
\end{lemma}
\begin{proof}
Every set of samples returned from the annealer $\mathcal{X}^{l}$ in iteration $l$ there are 2 cases:
\begin{enumerate}
\item Annealer returns 0 for all variables, for all samples. Then, no variables are pruned, and the while loop resumes.\label{all_zero}
\item There exists a sample $\mathcal{X}^{l}_m \in \mathcal{X}^{l} $ with at least one variable $x_{i,s,j,t}=1$. Then at least one variable will be in $\bar{X}^l$ since $\theta \in (0,1)$. \label{exists_one}
Then there are two cases: 
\begin{enumerate}
\item $\bar{X}^l$ does not contain a variable $x_{i,s,j,t}$ with either $i'=i$ or $i'=j$. Variables containing customers along paths returned by step \ref{find_paths} will be pruned according to the pruning rules in \S \ref{sec:pruning_rules}, and the while loop resumes its next iteration with $X^{l+1} \subset X^{l}$. \label{cust_not_selected} 
\item $\bar{X}^l$ contains a variable $x_{i,s,j,t}$ with either $i'=i$ or $i'=j$. \label{cust_selected} 
Then there are two cases: 
\begin{enumerate}
\item The paths returned by step \ref{find_paths} contain a tuple $(i',s)$ for some $s$. \label{contain}
\item They do not. \label{not_contain}
\end{enumerate} 
\end{enumerate}
\end{enumerate}
In case \ref{not_contain}, the variables containing customers along paths returned by step \ref{find_paths} will be pruned according to the pruning rules in \S \ref{sec:pruning_rules}, and the while loop resumes its next iteration with $X^{l+1} \subset X^{l}$.
If in case \ref{contain}. QED.
If the annealing time $A$ is sufficiently large, \cite{albash2018adiabatic} eventually, we will exit case \ref{all_zero}.
If we continue to be in cases \ref{cust_not_selected} or \ref{not_contain} above, then paths returned in step \ref{find_paths} in Alg. \ref{alg:greedy} do not contain tuple $(i',s)$, and in the next iteration $X^{l+1} \subset X^{l}$. By Lemma \ref{even_prune}, all variables in the interior of previously found paths $S$ will be pruned. At some iteration $l'$, the only variables not pruned are those containing tuple $(i',s)$, in which case they must appear in the paths returned in step \ref{find_paths}. 
\end{proof}
\begin{lemma} \label{even_interior}
If customer $i^R_k$ was in the exterior of path $R\in Q^l$ in iteration $l$, then it will be in the interior of some path $P\in S^{l'}$ for an iteration $l'>l$. 
\end{lemma}
\begin{proof}
If $k =0$, then all outgoing variables and incoming variables to $i^R_0$ with $t \neq s^R_0$ have been pruned in ``exterior'' rule \ref{exterior_first_not_depot}. So, the only active variables involving $i^R_0$ have $(j,t)=(i^R_0, s^R_0)$. By Lemma \ref{cust_appear}, in iteration $l'$ path $R' \in Q^{l'}$ will contain $(i^R_0, s^R_0)$. This can only happen if $i^{R'}_{|R'|} = i^R_0$. In step \ref{join_at_start} in Alg. \ref{alg:concat}, these paths are unioned, and the resulting path $P = R^{'\,^\frown} R$ will have $i^R_0$ in its interior. The proof for $k=|R|$ is the similar, except $i^{R'}_0 = i^R_{|R|}$, and the resulting path is $P = R ^\frown R'$.
\end{proof}
\begin{lemma} Alg. \ref{alg:greedy} converges to a feasible solution to Problem \eqref{eq:qubo_fsvrptw}. \label{feasible}
\end{lemma} 
\begin{proof} Observe that all variables contain a customer by pre-processing constraint \eqref{inv_cust}. By Lemma \ref{even_interior} at some iteration, any customer $i$ will be in the interior of some path $P \in S$. By the pruning rules, when a customer is in the interior of a path $P \in S$, we have set exactly two variables involving customer $i$ equal to $1$: one incoming and one outgoing. All other variables involving customer $i$ have been pruned and set to $0$. Thus, constraint \eqref{eq:qubo_cover} is satisfied. By the same argument, if $(i,s)$ is along a path in $S$ returned by Alg. \ref{alg:greedy}, then exactly one variable on each side of constraint (\ref{eq:flow},$i,s$) will be $1$ - all others have been pruned by rules ``interior'' \ref{outgoing} and \ref{incoming}. 
If $(i,s)$ is not along a path, then constraints (\ref{eq:flow}, $i,s$) will become $0=0$ and will be satisfied again by rules ``interior'' \ref{outgoing} and \ref{incoming}. Since constraints \eqref{eq:qubo_cover} and \eqref{eq:flow} are satisfied, the paths $S$ returned by Alg. \ref{alg:greedy} are a feasible solution to Problem \eqref{eq:qubo_fsvrptw}.
\end{proof}
\begin{lemma} If in Alg. \ref{alg:greedy} step \ref{step:annealing} is solved to optimality, and multiple longest paths are returned using the procedure described in section \ref{sec:greedy} then Alg. \ref{alg:greedy} converges to the optimal solution to Problem \eqref{eq:qubo_fsvrptw} in a single iteration. \label{exact_solver}
\end{lemma} 
\begin{proof} We give the proof for a sub-routine returning multiple longest paths as described in section \ref{sec:greedy}.  Observe that in the first iteration, step \ref{step:annealing} solves Problem \eqref{eq:qubo_fsvrptw} with all variables active, so it is equivalent to solving Problem \eqref{eq:qubo_fsvrptw}. If step \ref{step:annealing} is solved to optimality, then it returns a single sample, so the one-body expectation values are all $1$. If $\theta$ is taken to be a threshold strictly between $0$ and $1$, then $\bar{X}^l$ will only contain variables assigned a value of $1$ by the exact solver.  
Notice that by returning multiple longest paths as described in \S \ref{sec:greedy}, the collection of paths in $Q^l$ will contain all customers. 
Observe that step \ref{step:prune} will prune the variables in DAG($\bar{X}^l$) to a value of $1$, and any variables not in DAG($\bar{X}^l$) to a value of $0$, which is the same as the solution returned by the exact solver. 
Since $Q^l$ contains paths with all customers, all variables will be pruned, so Alg. \ref{alg:greedy} will terminate. From Lemma \ref{feasible}, the result will be feasible. Since an exact solver was used, it is also optimal.
\end{proof}

\begin{table}
 \centering
 \renewcommand{\arraystretch}{1.3}
 \caption{Notation and Solver descriptions}
 \begin{tabular}{|l|l|}
 \hline
 \textbf{Notation} & \textbf{Description} \\ \hline
 greedy+QPU & Alg. \ref{alg:greedy} with Step \ref{step:annealing} solved using \texttt{Advantage2} \\ \hline
 greedy+CQM & Alg. \ref{alg:greedy} with Step \ref{step:annealing} solved using \texttt{LeapHybridCQM} \\ \hline
 greedy+SAS & Alg. \ref{alg:greedy} with Step \ref{step:annealing} solved using \texttt{SimulatedAnnealingSampler} \\ \hline
 greedy+Exact & Alg. \ref{alg:greedy} with Step \ref{step:annealing} solved using \texttt{Gurobi} \\ \hline
 QPU & Problem \eqref{eq:qubo_fsvrptw} solved using \texttt{Advantage2} \\ \hline
 CQM & Problem \eqref{eq:qubo_fsvrptw} solved using \texttt{LeapHybridCQM} \\ \hline
 SAS & Problem \eqref{eq:qubo_fsvrptw} solved using \texttt{SimulatedAnnealingSampler} \\ \hline
 \end{tabular}
 \label{tab:notation}
\end{table}

\section{Computational Results} \label{sec:comp_results}
We benchmark with both classical and quantum annealing approaches for solving Problem \ref{eq:qubo_fsvrptw}.  
In Algorithm \ref{alg:greedy}, an annealer or a MIP solver may be used in step \ref{step:annealing}.  The solvers described in section \ref{sec:prelim} were used to either solve Problem \ref{eq:qubo_fsvrptw} or were used to solve step \ref{step:annealing} in Algorithm \ref{alg:greedy}.  The notation for each solution method is in Tab. \ref{tab:notation}. For the solution methods QPU, CQM, and SAS, a feasible solution was obtained by sorting samples by energy and selecting the first feasible solution.
To evaluate the performance of our proposed algorithm for practical-sized problems, Solomon's benchmark instances were used \cite{solomon1987algorithms}, briefly described as graphs with nodes that must be serviced within a time window, placed according to a distribution. Instance types $R101$/$R201$ indicate narrow/wide time windows resp. We modify an instance to suit our problem as follows: Randomly sample $N$ customers. For each value of $N$, construct ten random examples independently and aggregate their statistics.

\subsection{Results - greedy+CQM, SAS, and CQM} \label{sec:res_greedy_SAS_CQM} 
From Fig. \ref{fig:perc_optimal_r101}, and \ref{fig:perc_optimal_r201}, we can see that the greedy+CQM approach enjoys a small relative optimality gap as compared to SAS and CQM, in particular for large problem sizes $N$>15. This difference is even more pronounced in the $R$201 examples. For $N$=25, SAS returned feasible solutions for only 2 of the ten random examples, and $N$=50, it returned no feasible solutions. Statistics are aggregated only over feasible solutions.
Fig. \ref{fig:perc_speedup_r101}, and \ref{fig:perc_speedup_r201} show the relative time difference between the greedy algorithm and CQM/SAS. Positive values indicate that the greedy algorithm is faster. For large problem instances, greedy is always faster than SAS. And CQM is sometimes faster by a small margin for $N$>10. 

\begin{table}[H]
    \centering
    \renewcommand{\arraystretch}{1.3}
    \caption{\textbf{greedy+QPU Results:} $t$ is absolute time (sec), $\alpha$ is rel. optimality gap, $X$ is the number of logical/physical qubits.}
    \begin{tabular}{|l|l|l|l|l|l|l|l|l|}
    \hline
        \textbf{$N$} & \textbf{$t_{R101}$} & \textbf{$\alpha_{R101}$} & \textbf{$X_{log,R101}$} & \textbf{$X_{phys,R101}$} & \textbf{$t_{R201}$} & \textbf{$\alpha_{R201}$} & \textbf{$X_{log,R201}$} & \textbf{$X_{phys, R201}$} \\ \hline
        \textbf{5} & 0.08 & 0 & 44 & 415 & 0.071 & 0 & 24 & 356 \\ \hline
        \textbf{6} & 0.074 & 0 & 84 & 601 & 0.075 & 0 & 78 & 585 \\ \hline
        \textbf{7} & 0.08 & 10 & 109 & 890 & 0.08 & 10 & 101 & 792 \\ \hline
        \textbf{8} & 0.09 & 17 & 121 & 955 & 0.11 & 13 & 109 & 834 \\ \hline
        \textbf{9} & 0.09 & 33 & 127 & 1021 & 0.09 & 27 & 128 & 998 \\ \hline
    \end{tabular}
    \label{tab:qpu_results}
\end{table}
\begin{figure}[H]
 \begin{subfigure}{0.48\textwidth}
 \includegraphics[width=\linewidth]{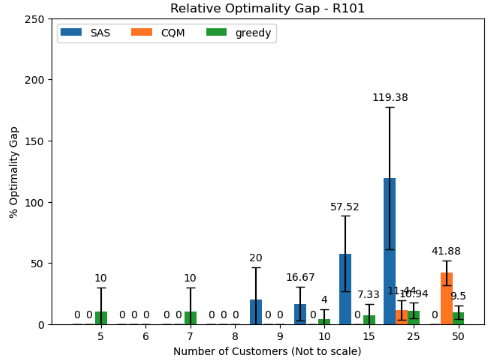}
 \caption{Solomon R101 Examples} \label{fig:perc_optimal_r101}
 \end{subfigure}\hfill
 \begin{subfigure}{0.49\textwidth}
 \includegraphics[width=\linewidth]{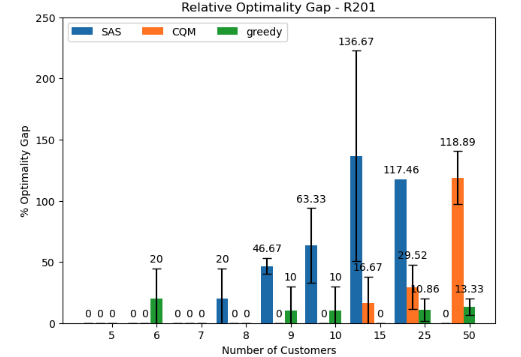}
 \caption{Solomon R201 Examples} \label{fig:perc_optimal_r201}
 \end{subfigure}%
\caption{Average relative optimality gap for SAS, and CQM and greedy+CQM, averaged over 10 examples with standard deviation. The relative optimal value is computed as $\frac{C_{A}-C_{opt}}{C_{opt}}$ where $C_A$ is the objective value found by the annealing method, and $C_{opt}$ is the optimal value found by \texttt{Gurobi}. Statistics are computed only over feasible solutions. Data not shown for SAS, $N$=50 - no feasible solutions found.}
\end{figure}

\begin{figure}[H]
\centering
 \begin{subfigure}{0.495\textwidth}
 \includegraphics[width=\linewidth]{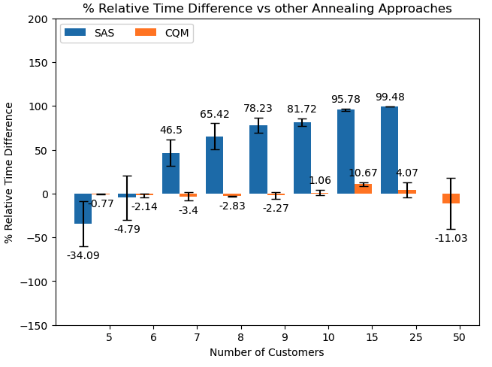}
 \caption{Solomon R101 Examples} \label{fig:perc_speedup_r101}
 \end{subfigure}%
 \begin{subfigure}{0.495\textwidth}
 \includegraphics[width=\linewidth]{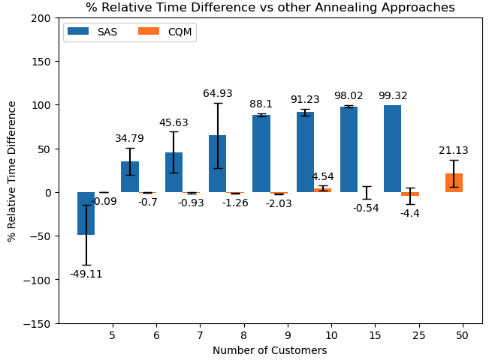}
 \caption{Solomon R201 Examples} \label{fig:perc_speedup_r201}
 \end{subfigure}%
\caption{Relative time difference between greedy+CQM and SAS /CQM averaged over ten examples. Rel. time difference is computed as $\frac{t_{A}-t_{greedy}}{t_{A}}$ where $t_A$ is the time of SAS/CQM, and $t_{greedy}$ is  time taken of greedy+CQM. Statistics are computed only over feasible solutions. Data not shown for SAS, $N$=50, because no feasible solutions found.}
\end{figure}

\begin{table}[H]
    \centering
    \renewcommand{\arraystretch}{1.3}
    \caption{\textbf{Results for greedy+QPU for R101 and R201}; $\alpha$ is \% rel. optimality gap, $v$ is \% valid solutions, $t$ is embedding time}
    \begin{tabular}{|l|l|l|l|l|l|l|l|l|l|l|}
    \hline
        \textbf{$N$} & \textbf{$\alpha_{R101}$} & \textbf{$v_{greedy,R101}$} & \textbf{$t_{emb,R101}$} & \textbf{$\alpha_{R201}$} & \textbf{$v_{greedy,R201}$} & \textbf{$t_{emb,R201}$} \\ \hline
        \textbf{5} & 0.0 $\pm$ 0.0 & 100 & 0.09 $\pm$ 0.01 & 0.0 $\pm$ 0.0 & 100 & 2.31 $\pm$ 1.40 \\ \hline
        \textbf{6} & 10.0 $\pm$ 2.1 & 100 & 2.44 $\pm$ 1.5 & 8.0 $\pm$ 4.1 & 100 & 0.16 $\pm$ 0.04 \\ \hline
        \textbf{7} & 2.5 $\pm$ 1.4 & 100 & 1.8 $\pm$ 0.3 & 12.4 $\pm$ 3.5 & 80 & 9.85 $\pm$ 6.66 \\ \hline
        \textbf{8} & 0.0 $\pm$ 0.0 & 40 & 2.7 $\pm$ 1.1 & 7.0 $\pm$ 4.2 & 20 & 0.16 $\pm$ 0.0 \\ \hline
        \textbf{9} & NaN & 0 & NaN & NaN & 0 & NaN \\ \hline
    \end{tabular}
    \label{tab:greedy_DWaveSampler}
\end{table}

\begin{table}[H]
 \centering
 \renewcommand{\arraystretch}{1.3}
 \caption{Computation Time (sec) - greedy+SAS}
 \begin{tabular}{|l|l|l|l|l|}
 \hline
 \textbf{$N$} & \textbf{$t_{\text{R101,greedy+SAS}}$} & \textbf{$t_{\text{R101,CQM}}$} & \textbf{$t_{\text{R201,greedy+SAS}}$} & \textbf{$t_{\text{R201,CQM}}$} \\ \hline
 \textbf{5} & 3.29 & 5.09 & 4.12 & 5.11 \\ \hline
 \textbf{6} & 6.03 & 5.06 & 9.01 & 5.14 \\ \hline
 \textbf{7} & 10.32 & 5.07 & 13.64 & 5.21 \\ \hline
 \textbf{8} & 19.51 & 5.2 & 25.14 & 5.33 \\ \hline
 \textbf{9} & 41.1 & 5.3 & 65.05 & 5.45 \\ \hline
 \textbf{10} & 51.14 & 5.6 & 93.61 & 5.63 \\ \hline
 \textbf{15} & 230.58 & 6.77 & 618.08 & 6.74 \\ \hline
 \textbf{25} & 21.6 & 10.96 & 43.1 & 10.93 \\ \hline
 \textbf{50} & 433.16 & 65.88 & 744.42 & 104.67 \\ \hline
 \end{tabular}
 \label{tab:greedy_SAS}
\end{table}

\begin{figure}[H]
 \begin{subfigure}{0.45\textwidth}
 \includegraphics[width=\linewidth]{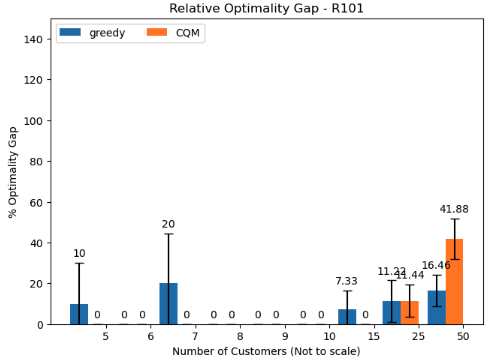}
 \caption{Solomon R101 Examples} \label{fig:greedy_SAS_perc_optimal_r101}
 \end{subfigure}%
 \begin{subfigure}{0.45\textwidth}
 \includegraphics[width=\linewidth]{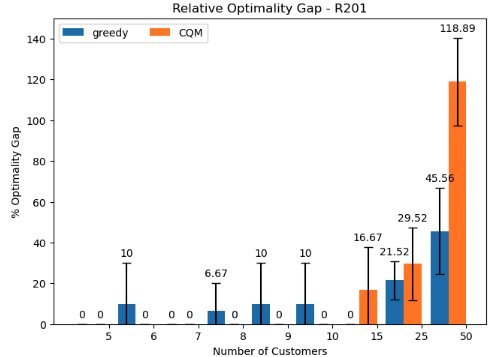}
 \caption{Solomon R201 Examples} \label{fig:greedy_SAS_perc_optimal_r201}
 \end{subfigure}%
\caption{Average relative optimality gap comparing greedy+SAS and CQM averaged over ten examples with standard deviation.}
\end{figure}

\subsection{Results - greedy+SAS vs CQM}
To show the effectiveness of our decomposition, greedy+SAS and CQM are also compared. The relative optimality gap for various problem sizes is shown in Fig. \ref{fig:greedy_SAS_perc_optimal_r201}, and the computation times are shown in Table \ref{tab:greedy_SAS}. At small problem sizes, the rel. The optimality gap remains low, and at large problem sizes, the optimality gap for the greedy algorithm is significantly lower than that of CQM. This trend is even more pronounced for the $R$201 examples, showing that the decomposition performs better than solving Problem \eqref{eq:qubo_fsvrptw} using a state-of-the-art hybrid solver CQM.

\subsection{Results - greedy+QPU vs CQM}
In Tab. \ref{tab:greedy_DWaveSampler}, we can see results averaged over ten examples for solving greedy+QPU. Because the process of embedding the logical problem on the physical device is non-deterministic, the following procedure was used for each problem size: For each of the ten random examples, the embedding process was executed 10 times. The embedding using the fewest number of physical qubits was used to execute the remaining iterations of Algorithm \ref{alg:greedy}.  The average embedding time was recorded and added to the time taken to complete the remaining iterations of the algorithm. The statistics for optimality gap and computation time are in Tab. \ref{tab:greedy_DWaveSampler}. For problem sizes $N \in \{5,6,7\}$, a valid embedding was found for all ten random examples. For $N=8$, $40\%$/$20\%$ of the examples, a valid embedding was found for $R$101/$R$201, respectively, and for $N=9$, none of the examples found a valid embedding. Given the current QPU capabilities, a hybrid solver for the QUBO subproblems is necessary for practical-scale problems. For the same examples, CQM always found the optimal solution. 

\subsection{Results - greedy+QPU vs QPU}

To evaluate robustness to noise, greedy+QPU and QPU were compared. For a fair benchmark, the parameters of the \texttt{DWaveSampler} were tuned in pure QPU. An annealing pause was used by setting the schedule to \texttt{[0.0, 0.0], [10.0, 0.4], [50.0, 0.4], [120, 1.0]}. Annealing pauses are shown to perform well for scheduling problems \cite{marshall2019power}. Finally, to increase the chance of the \texttt{DWaveSampler} finding a feasible solution, both the annealing time and the number of samples were adjusted to the largest values allowed by D-Wave's QPU access time: $50 \mu s$ and $M$=2000, respectively. In  greedy+QPU, we left the annealing parameters in \texttt{DWaveSampler} to their defaults. 
Even after these enhancements \texttt{DWaveSampler} still fails to find a feasible solution in any of the ten examples for problem sizes greater than N=5 (with/without pausing). For N=5, only 20\% of the solutions were feasible (not shown in Tab. \ref{tab:qpu_results}). With pausing \cite{marshall2019power}, surprisingly, this probability decreased to 10\%. However, greedy+QPU still maintains an average relative optimality gap below 33\%. The total time of Alg. \ref{alg:greedy} and the average relative optimality gap are shown in Tab. \ref{tab:qpu_results}.  The average number of logical/physical qubits is also reported ($X$) to show the scaling of this formulation on the binary variables.  As $N$ increases, the number of logical qubits begins to attenuate. Only one discretization value within each time window is needed. As a timetable includes more customers, more of the time windows overlap - avoiding the need to create additional time discretizations.

\section{Conclusions and Future Work}
This work presents the first effective algorithm for generating routes using an annealing-based approach. We demonstrated its efficiency by solving the Fleet Sizing Vehicle Routing Problem with Time Windows (FSVRPTW). We prove that it converges to a feasible solution, and we benchmark it against state-of-the-art classical and hybrid annealing-based approaches. We showed that it enjoys a smaller optimality gap than other annealing-based approaches at a competitive computation time, even for practical-sized problems. 
It also shows to be noise-robust when executed on a QPU, taking advantage of the entire sample set returned by the annealer. 
Natural directions for future work include determining the dependence of the proposed algorithm's performance on the dataset, e.g., using different time discretizations and graphs generated from other distributions. We may also explore the effect of the control parameters, e.g., $\theta$. These future directions will provide insight into how this adaptive algorithm may be used to solve a more general class of problems outside of routing and scheduling.

\section{Acknowledgements}
This work was supported in part by the PWICE fellowship of USC.

\bibliographystyle{unsrt}  
\bibliography{references}

\end{document}